\documentclass[pra,twocolumn,showpacs,floatfix,aps]{revtex4-2}
\usepackage{amsmath,amssymb,epsfig,amsfonts,bbm,graphicx,times,palatino}
\usepackage{subfigure}
\usepackage{color}
\usepackage[latin1]{inputenc}
\usepackage[english]{babel}
\usepackage{url}
\usepackage{textcase}
\usepackage{bm}
\usepackage{hyperref}
\usepackage{easybmat}
\usepackage{dsfont}
\usepackage{amsthm}
\usepackage{braket}
\usepackage{ulem}

\newtheorem{theorem}{Theorem}
\newtheorem{corollary}{Corollary}[theorem]

\begin{document}
\normalem

\title{Genuine multipartite entanglement is not a precondition for secure conference key agreement}
\author{Giacomo Carrara}
\email{carrara@uni-duesseldorf.de}
\author{Hermann Kampermann}
\author{Dagmar Bru\ss}
\author{Gl\'aucia Murta}
\email{glaucia.murta@uni-duesseldorf.de}
\affiliation{Institut f\"ur Theoretische Physik III, Heinrich-Heine-Universit\"at D\"usseldorf, Universit\"atsstra\ss{}e 1, D-40225 D\"usseldorf, Germany}

\begin{abstract}
Entanglement plays a crucial role in the security of quantum key distribution. A secret key can only be obtained by two parties if there exists a corresponding entanglement-based description of the protocol in which entanglement is witnessed, as shown by Curty \textit{et al} (2004) \cite{Curty2004}. Here we investigate the role of entanglement for the generalization of quantum key distribution to the multipartite scenario, namely conference key agreement. In particular, we ask whether the strongest form of multipartite entanglement, namely genuine multipartite entanglement, is necessary to establish a conference key. We show that, surprisingly, a non-zero conference key can be obtained even if the parties share biseparable states in each round of the protocol. Moreover we relate conference key agreement with entanglement witnesses and show that a non-zero conference key can be interpreted as a non-linear entanglement witness that detects a class of states which cannot be detected by usual linear entanglement witnesses.
\end{abstract}

\maketitle

\noindent\textit{Introduction --} Secure communication is a central demand for modern society. Security can be provided by Quantum Key Distribution (QKD) which readily enters the industrial market. In QKD \cite{BB84,E91} entanglement plays a crucial role in the security proofs~\cite{BB84proof,LoProof}. Indeed, even prepare-and-measure protocols~\cite{BB84,sixstate}, which do not require any entanglement for their implementation, have an entanglement-based counterpart~\cite{BBM92} which can be used for the protocol's security analysis. In Ref.~\cite{Curty2004}, the authors showed that entanglement is in fact a necessary condition to obtain a secure key in a QKD protocol and, moreover, the entanglement of the state shared by Alice and Bob can be witnessed using the measurements performed in the protocol.

We consider a generalization of QKD to the scenario where $N$ parties wish to establish a common shared secret key. This task is called conference key agreement (CKA) and allows for secure broadcast. CKA can be achieved using a concatenation of bipartite QKD \cite{bilinks1,bilinks2, bilinks3}, together with additional classical communication. However, the rich structure of multipartite correlations opens the possibility to design new protocols which can have clear advantages in certain network architectures \cite{Epping2017}. Several protocols exploiting the correlations of multipartite entangled  states have been proposed using qubit systems in the device-dependent \cite{cabello2000,chen2004,Epping2017,Federico2018,Federico2019} and device-independent scenario \cite{DI2018,DI2019,TimoComment}, as well as continuous-variables systems \cite{CV2016,Zhang2018,Ottaviani2019}. Even a proof of principle implementation of CKA with four nodes has been recently realized \cite{CKAexperiment}.

Here we ask the question of whether the strongest form of multipartite entanglement, namely genuine multipartite entanglement, is a necessary ingredient for CKA based on multipartite quantum correlations. We will show that, counter-intuitively, this is not the case: $N$ parties can establish a secret conference key even when the state distributed in each round of the protocol is biseparable. Moreover, we prove that, in order to obtain a non-zero conference key, the measurements used in the protocol need to be able to witness entanglement across any partition of the set of parties, extending the result of Ref. \cite{Curty2004} to the multipartite scenario.

\noindent{\textit{Preliminaries --}} We focus on CKA protocols \cite{CKAreview} consisting of several rounds where, in each round, a single copy of a multipartite state is distributed to the $N$ parties, namely Alice and Bob$_1$, \ldots, Bob$_{N-1}$. Upon receiving the systems, the parties perform local measurements and record the classical outcome.

In such protocols, an important figure of merit is the asymptotic secret key rate, i.e. the ratio between the number of extracted secret bits and the number of shared copies of the state, in the limit of an infinite number of rounds. Analogously to the bipartite case \cite{Renner2005,devetak}, the asymptotic secret key rate of the CKA protocols under consideration can be expressed, after the usual post-processing (parameter estimation, one-way information reconciliation and privacy amplification) as \cite{Epping2017}
\begin{equation}\label{keymulti}
r^{\infty}= \max{ \left[ 0, H(X|E) - \max_{i} H(X|Y_i) \right]},
\end{equation}
where $X$ and $Y_i$ 
denote the registers that store the outcomes of the measurements performed by Alice and Bob$_i$, respectively, in the key generation rounds. Here $H(X|E)=H(XE)-H(E)$ is the von Neumann entropy of Alice's outcome in the key generation rounds, conditioned on Eve's (possibly quantum) side information. $H(X|Y_i)=H(XY_i)-H(Y_i)$ represents the amount of information Alice needs to communicate to Bob$_i$ so that he can correct his raw key. The maximum over the Bobs in Eq. \eqref{keymulti} illustrates the fact that Alice needs to communicate enough information to correct for the worst case of the Bobs. We recall that for a state $\rho_X$ of a system $X$, the quantum von Neumann entropy is defined as $H(X)=-\mbox{Tr}[\rho_X \log \rho_X]$.

The conditional von Neumann entropy satisfy the following properties \cite{tom}:
\begin{enumerate}
\item Additivity for product states \cite[Corollary 5.9]{tom}: if $\rho_{AB}=\rho_A \otimes \rho_B$  then $H(A|B) = H(A)$.
\item Data-processing \cite[Corollary 5.5]{tom}: considering $\rho_{ABC}$ then $H(A|BC) \leq H(A|B)$. 
\item Conditioning on classical information \cite[Proposition 5.4]{tom}: if $\rho_{ABF}=\sum_j q_j \rho_{AB}^j \otimes | j \rangle \langle j |_F$ is a classical-quantum state where the system $F$ is a classical register, then $H(A|BF) = \sum_j q_j H(A | B F=j)$ where $H(A | B F=j)$ is evaluated on the state $\rho_{AB}^j$.
\end{enumerate}

Our goal is to investigate the role of multipartite entanglement in the single copy of the state shared by the $N$ parties in each round of the protocol. In the bipartite case either the state is separable and no key can be extracted, or the state is entangled and can potentially be used for QKD~\cite{Curty2004}. In the multipartite scenario, however, different classes of entanglement can be defined, which have been extensively studied \cite{EntHoro,walter2016multipartite,criteria2017,Entrev2019,distreview}. 

Let $S_{\alpha}$ be a proper subset of the parties and $\bar{S}_{\alpha}$ be the complement. Then a state $\rho_{AB_1 \dots B_{N-1}}$ is \emph{separable with respect to the partition} $S_{\alpha}|\bar{S}_{\alpha}$ if it is of the form
\begin{equation}
\label{onesep}
\rho_{AB_1 \dots B_{N-1}}= \sum_j q_j \rho^j_{S_{\alpha}} \otimes \rho^j_{\bar{S}_{\alpha}},
\end{equation}
where $\rho^j_{S_{\alpha}}$ and $\rho^j_{\bar{S}_\alpha}$ are states shared by the parties in $S_\alpha$ and $\bar{S}_\alpha$, respectively, and where $q_j \geq 0$ and $\sum_j q_j =1$.

A state is called biseparable \cite{EntHoro}, if it is a convex combination of states that are separable with respect to different partitions, that is
\begin{equation}\label{bs}
\rho_{bs}=\sum_{S_\alpha} \sum_j q_{S_\alpha}^j \rho^j_{S_{\alpha}} \otimes \rho^j_{\bar{S}_\alpha},
\end{equation}
where the first sum is performed over all proper subsets $S_\alpha$ of the parties. Again, the coefficients must satisfy $q_{S_{\alpha}}^j \geq 0 \; \forall j, S_{\alpha}$ and $\sum_{\alpha} \sum_j q_{S_{\alpha}}^j =1$. It is worth noting that a state can be biseparable, yet not separable with respect to any partition.

Finally, if a state cannot be written in the form of Eq. \eqref{bs} we call it \emph{genuine multipartite entangled (GME)}. All CKA protocols based on multipartite entanglement proposed so far \cite{cabello2000,chen2004,Epping2017,Federico2018,Federico2019,DI2018,DI2019,CV2016,Zhang2018,Ottaviani2019}, explore the correlations of GME states, such as the Greenberger-Horne-Zeilinger (GHZ) state \cite{GHZ} or the W state \cite{W}.

\noindent{\textit{Entanglement is necessary for CKA --}} In the following we prove that entanglement across all partitions in the state shared by the parties is necessary in order to lead to a non-zero asymptotic conference key rate.
\begin{theorem}
\label{teo}
Given a CKA protocol, if the state shared by the $N$ parties is separable with respect to some partition $S_{\alpha}|\bar{S}_{\alpha}$,  then $r_{\infty} = 0$.
\end{theorem}

\begin{proof}[Proof of Theorem~\ref{teo}]
To prove the statement, since the asymptotic key rate in Eq. \eqref{keymulti} includes an optimization over all the Bobs, it suffices to prove that $H(X|Y_l) \geq H(X|E)$ for a specific Bob$_l$. Let us consider a state separable with respect to a partition $S_{\alpha}|\bar{S}_{\alpha}$, in the form of Eq. \eqref{onesep}, such that $S_\alpha$ contains Alice. We consider a Bob contained in $\bar{S}_{\alpha}$, let us say Bob$_{l}$. Let Eve have a purification of the state of the form
\begin{equation}
|\psi_{AB_1,\dots,B_{N-1}EFF'} \rangle = \sum_{j} \sqrt{q_j} |\psi^j_{S_\alpha \bar{S}_\alpha E} \rangle |j\rangle_F |j \rangle_{F'}, 
\end{equation} 
where $|\psi^j_{S_\alpha \bar{S}_\alpha E} \rangle$  is a  purification of $\rho_{S_{\alpha}}^j \otimes \rho^j_{\bar{S}_\alpha}$ and the systems $F$ and $F'$ are classical registers held by Eve. The additional classical register $F'$ is necessary to exploit the properties of the von Neumann entropy of classical-quantum states. In fact, tracing out the system $F'$, Eve's system $E$ and all the Bobs except $B_l$ will result in a state of the form 
\begin{equation}
\label{cq}
\rho_{AB_lF}= \sum_j q_j \rho^j_{A} \otimes \rho_{B_l}^j \otimes |j \rangle \langle j |_F
\end{equation}
which is a classical-quantum state consisting of a separable state for Alice and Bob $B_l$, paired with the classical register $F$ held by Eve. We remark that performing local measurements on a separable state will result in a separable state. Thus, after the measurements of the CKA protocol the state will still be in the form of Eq. \eqref{cq}. Moreover, we can write the following chain of inequalities:
\begin{align}
 \begin{split}  
H(X|Y_l) & \geq  H(X|Y_l F)   \\
& =  \sum_j q_j H(X | Y_l F=j)   \\
& =  \sum_j q_j H(X | F=j)  \\
& =  H( X | F) \geq H(X| EFF') = H(X |E_{tot}) 
\end{split}
\end{align}
where $E_{tot}$ indicates the global subsystem of Eve, which includes the classical registers. In the first, second and third line we used Property 2, Property 3 and Property 1 of the conditional Von Neumann entropy, respectively. Finally, in the fourth line we used again Properties 2 and 3. This concludes the proof. 
\end{proof}
It follows that there must be some entanglement shared between Alice and all the Bobs in order to establish a secret common key. It is worth noting that for $N=2$ this proof simplifies the argumentation given in Ref. \cite{Curty2004}.

\noindent{\textit{CKA without GME --}} We will now focus on the main question, that is whether a positive conference key can be established without GME. We answer this question in the affirmative by exhibiting a family of biseparable states that can lead to non-zero conference key:
\begin{align}
\begin{split}\label{state}
\rho_{AB_1, \dots, B_{N-1}}^{(N,k)}=&\\
\sum_{\underset{S_\alpha \in \mathcal{S}^{(k)} }{\alpha}} \frac{1}{\mathcal{N}}& \Phi^{GHZ,k}_{S_{\alpha}} \bigotimes_{\underset{B_m \in \bar{S}_{\alpha} }{m}} | + \rangle \langle + |_{B_m},
\end{split}
\end{align}
where $\mathcal{S}^{(k)}$ is the set of subsets of $k$ parties that contain Alice and $k-1$ Bobs, $\Phi^{GHZ,k}_{S_{\alpha}}=|GHZ \rangle\langle GHZ |_{S_{\alpha}}$ is the projector of the GHZ state shared by the $k$ parties of the subset $S_\alpha$, defined as $|GHZ \rangle_{S_\alpha}=\frac{1}{\sqrt{2}} \left( |0 \rangle^{\otimes k} + | 1 \rangle^{\otimes k}\right)$ and $|+\rangle = \frac{1}{\sqrt{2}}\left(|0 \rangle + |1 \rangle \right)$. The normalization factor is equal to $\mathcal{N}=\binom{N-1}{k-1}$ since the number of terms in the convex combination is equal to the number of subsets of cardinality $k-1$ within the $N-1$ Bobs.

We show that this family of states can be used to generate a non-zero key in a simple conference key agreement protocol, namely the N-BB84 protocol \cite{Federico2018}. The N-BB84 protocol consists of $X$-basis measurements for the parameter estimation rounds and $Z$-basis measurements for the key generation rounds.

The asymptotic conference key rate of the N-BB84 protocol for the family of states $\rho_{AB_1, \dots, B_{N-1}}^{(N,k)}$, Eq. \eqref{state}, as a function of the total number of parties $N$ and the number of parties $k$ that are entangled is given by:
\begin{eqnarray}\label{keyrate}
r_{\rm N-BB84}^{\infty}(N,k) & = & \frac{1}{2}\frac{N-k}{N-1}\log_2{\left(\frac{N-k}{N-1}\right)}+ \nonumber \\
& + & \frac{1}{2}\frac{N+k-2}{N-1}\log_2{\left(\frac{N+k-2}{N-1}\right)}.
\end{eqnarray}
A detailed derivation of $r_{\rm N-BB84}^{\infty}(N,k)$ is presented in the Supplemental Material. There, we also show that the key rate given in Eq.~\eqref{keyrate} is  optimal for the family of states \eqref{state},  when the key is generated with measurements in the $Z$ basis.

In Figure \ref{fN} we show the secret key rate as a function of the number of parties $N$ for different values of the number of entangled parties $k$. For comparison, we also plot the key rate of a CKA protocol based on the concatenation of multiple bipartite QKD protocols, in the noiseless scenario, for a network with bottleneck \cite{Epping2017}. In this case, Alice runs $N-1$ bipartite QKD protocols in order to establish a secret key with each of the Bobs.
\begin{figure}[!htbp]
\centering
\includegraphics[height=7.2cm, width=7.2cm, keepaspectratio]{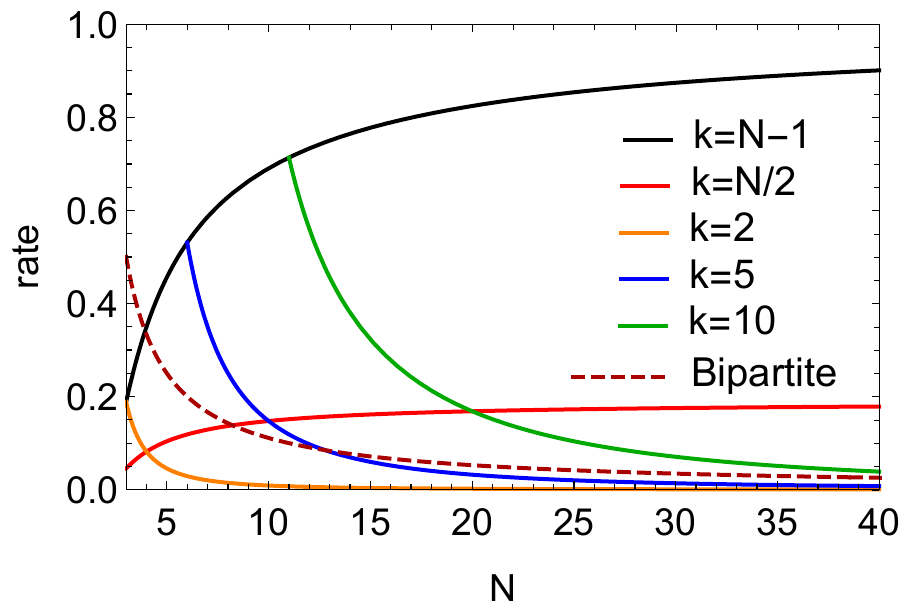}
\caption{Asymptotic secret key rate for the state of Eq. \eqref{state} for different 
values of k (straight lines) and key rate of CKA based on multiple noiseless bipartite QKD protocols (dashed line), both as a function of N. We remark that since $ k \leq N-1$, the curves for fixed $k$ start at different values of $N$.}
\label{fN}
\end{figure}

Figure \ref{fN} shows that $r_{\rm N-BB84}^{\infty}$ approaches $1$ as $N$ increases, if $k$ equals $N-1$. Moreover, even in the case when only $2$ parties, Alice and one of the Bobs, are entangled in each term of the mixture, a non-zero secret key can be obtained. However, for a fixed value of $k$,  $r_{\rm N-BB84}^{\infty} \to 0 $ as $N$ increases. The comparison with the key rate of a concatenation of multiple bipartite QKD protocols yields interesting results: while, on one hand, no advantage can be obtained for $k=2$, on the other hand an advantage can  be obtained in the regime of a  $k$ close to $N$, with a marked advantage for high $k$.

To further analyze the advantage obtainable with the presented protocol compared to the concatenation of bipartite QKD protocols, we evaluate the performance of the family of states \eqref{state} in the presence of noise. We consider the case where the qubit of each Bob undergoes a \emph{local depolarizing channel} $\mathcal{D}$, where $\mathcal{D}[\rho]=(1-p)\rho + p \frac{\mathds{1}}{2}$. We compare this with a concatenation of bipartite QKD protocols that undergo the same type of noise. Details of this analysis can be found in the Supplemental Material. Figure \ref{noisemain} illustrates the result for $N$=6.
\begin{figure}[!htbp]
\centering
\includegraphics[height=7.2cm, width=7.2cm, keepaspectratio]{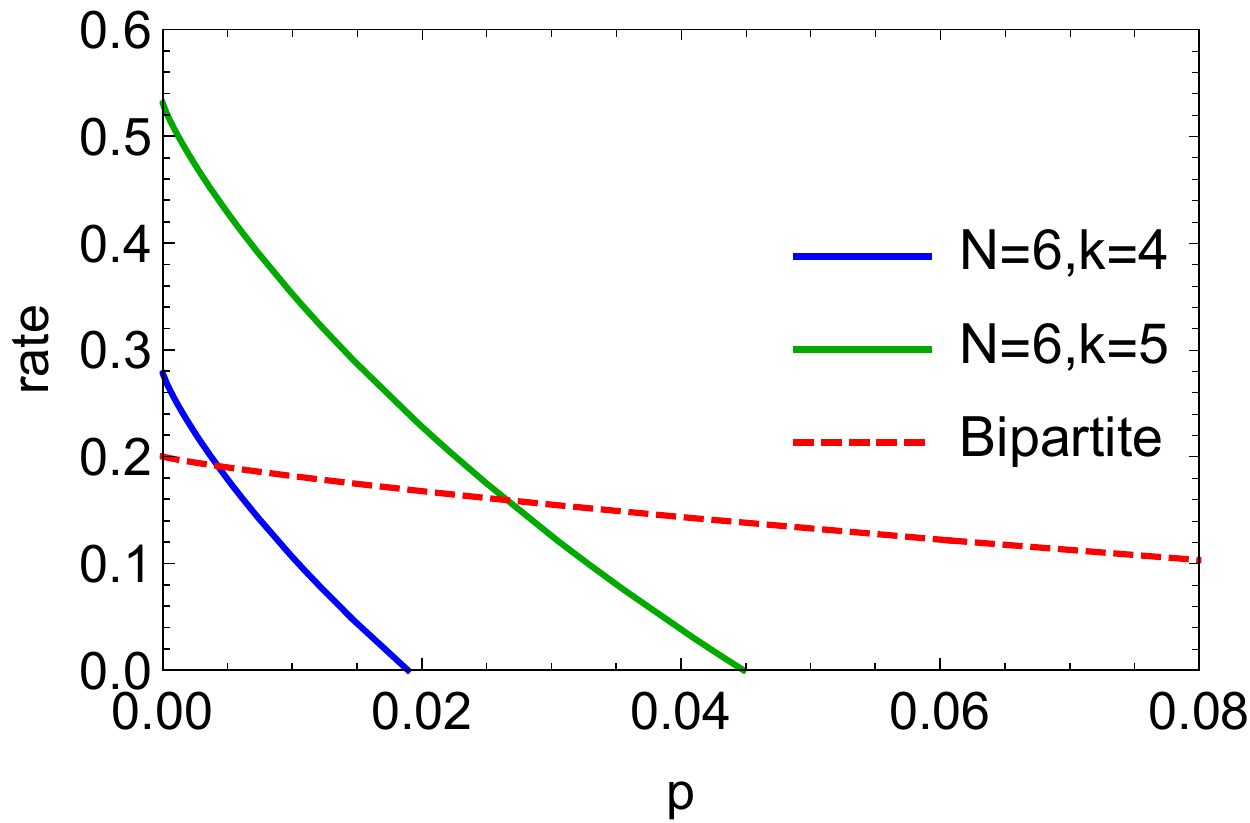}
\caption{Plot of the asymptotic key rate of the N-BB84 protocol for the state of Eq. \eqref{state} undergoing local depolarizing noise (solid lines), as a function of the depolarizing channel parameter $p$, for fixed $N=6$ and different $k$: $k=4$ (blue, left) and $k=5$ (green, right). The results are compared with the key rate of a concatenation of noisy bipartite BB84 QKD protocols (red dashed line).}
\label{noisemain}
\end{figure}
Even in the noisy scenario, an advantage can be obtained in the low noise regime and for $k$ close to $N$.

Our results show that CKA without GME states is possible. We remark that in Ref.~\cite{das2019universal} the authors have established that GME is a necessary condition for non-zero key in a one-shot  conference key agreement protocol. This result, at first, seems in contradiction to our findings, however Ref.~\cite{das2019universal}  refers to the \emph{global} input state, that for the class of protocols we consider would be $\rho_{AB_1 \dots B_{N-1}}^{\otimes n}$, where $n$ is the number of rounds. Since the set of biseparable states is not closed under tensor product, the global input state can be GME even if the single copy of the state is biseparable. Here we focus on analysing the entanglement properties of the single copy of the states. This is because we consider a class of protocols in which the states are distributed and measured at each round, therefore no storage or quantum global operation on all the copies is required.

\noindent\emph{CKA and entanglement witnesses --} Theorem \ref{teo} provides us with a necessary condition to obtain a non-zero key rate in a CKA protocol. We now want to extend to the multipartite scenario the bipartite result presented in Ref. \cite{Curty2004}: no secret key can be extracted in a QKD protocol unless Alice and Bob are able to witness entanglement in the shared state using the measurements performed in the protocol. An entanglement witness \cite{witness2,witness3,distreview} is a Hermitian operator $W$ such that $\mbox{Tr}(W\sigma) \geq 0$ for all separable states $\sigma$ and $\mbox{Tr}(W\rho) < 0$ for at least one entangled state $\rho$. This definition of an entanglement witness is based on the fact that the set of separable states is closed and convex, and can thus be separated with a hyperplane from its complement \cite{HB,witness2}. 
In the multipartite scenario, given the more intricate structure of possible correlations, witnesses can be defined to distinguish different classes of states \cite{distreview}.
We thus consider the same approach of Ref. \cite[Theorem 1]{Curty2004}: starting from the measurements performed by the parties, we analyze the entanglement witnesses that can be constructed with them. We obtain the following theorem.

\begin{theorem}
\label{wit}
Given a CKA protocol in which the parties use a set of local measurements, for the test  and key generation rounds, which are represented by the POVMs $\{G^a_{x}\}, \{G^{b_1}_{y_1}\}, \dots, \{G^{b_{N-1}}_{y_{N-1}}\}$, where $a,b_1, \dots b_{N-1}$ indicate the outputs of the measurements labeled by $x,y_1,\dots,y_{N-1}$, 
then one can obtain a non-zero asymptotic conference key rate $r_\infty > 0$ only if the presence of entanglement can be proved across any partition of the parties into two subsets.

Moreover, the presence of entanglement across each bi-partition can be verified through a set of entanglement witnesses of the form
\begin{equation}
W_\alpha=\sum_{\underset{a, b_1, \dots , b_{N-1}}{x,y_1,\dots,y_{N-1}}} c_{\underset{a, b_1, \dots , b_{N-1}}{x,y_1,\dots,y_{N-1}}}^{(\alpha)} G_{x}^{a} \otimes G^{b_1}_{y_1} \otimes \dots \otimes G^{b_{N-1}}_{y_{N-1}}
\end{equation} 
where $\alpha$ labels the partition $S_\alpha | \bar{S}_\alpha$ with $S_\alpha$ being a proper subset of the parties and $\bar{S}_\alpha$ is its complement, and where $c_{\underset{a, b_1, \dots , b_{N-1}}{x,y_1,\dots,y_{N-1}}}^{(\alpha)}$ are real coefficients.
\end{theorem}
The proof is given in the Supplemental Material. Theorem~\ref{wit} implies that entanglement across any bi-partition can be witnessed using the statistics of results of the measurements specified by the protocol, since the witness operators $W_{\alpha}$ are constructed from the POVM elements of these measurements. Theorem~\ref{wit}, combined with the results of the previous Section, leads to the following Corollary.
\begin{corollary}
The figure of merit $r_\infty > 0$ is a non-linear entanglement witness, detecting the presence of entanglement across any bi-partition of the parties.
\end{corollary}
This corollary is due to the result of Theorem \ref{wit} in combination with the examples presented in the previous Section: In fact, the union of all the sets of states that are separable with respect to a specific partition is not a convex set and thus cannot be separated by linear witnesses from its complement \cite{witness2} (see Figure \ref{sepset}). Moreover, if a CKA protocol is performed and a non-zero key rate is obtained, it is a necessary condition that the state shared by the parties is not separable across any partition of the parties. 
Therefore, a non-zero key rate reveals that the state utilized in the protocol is outside of the union of the sets of states that are separable with respect to a fixed partition. Finally, the results of the previous Section tell us that non-GME states can also lead to a non-zero conference key, thus allowing us to conclude that the witness cannot be linear, hence the corollary. 

\begin{figure}[!htbp]
\centering
\includegraphics[height=7.2cm, width=7.2cm, keepaspectratio]{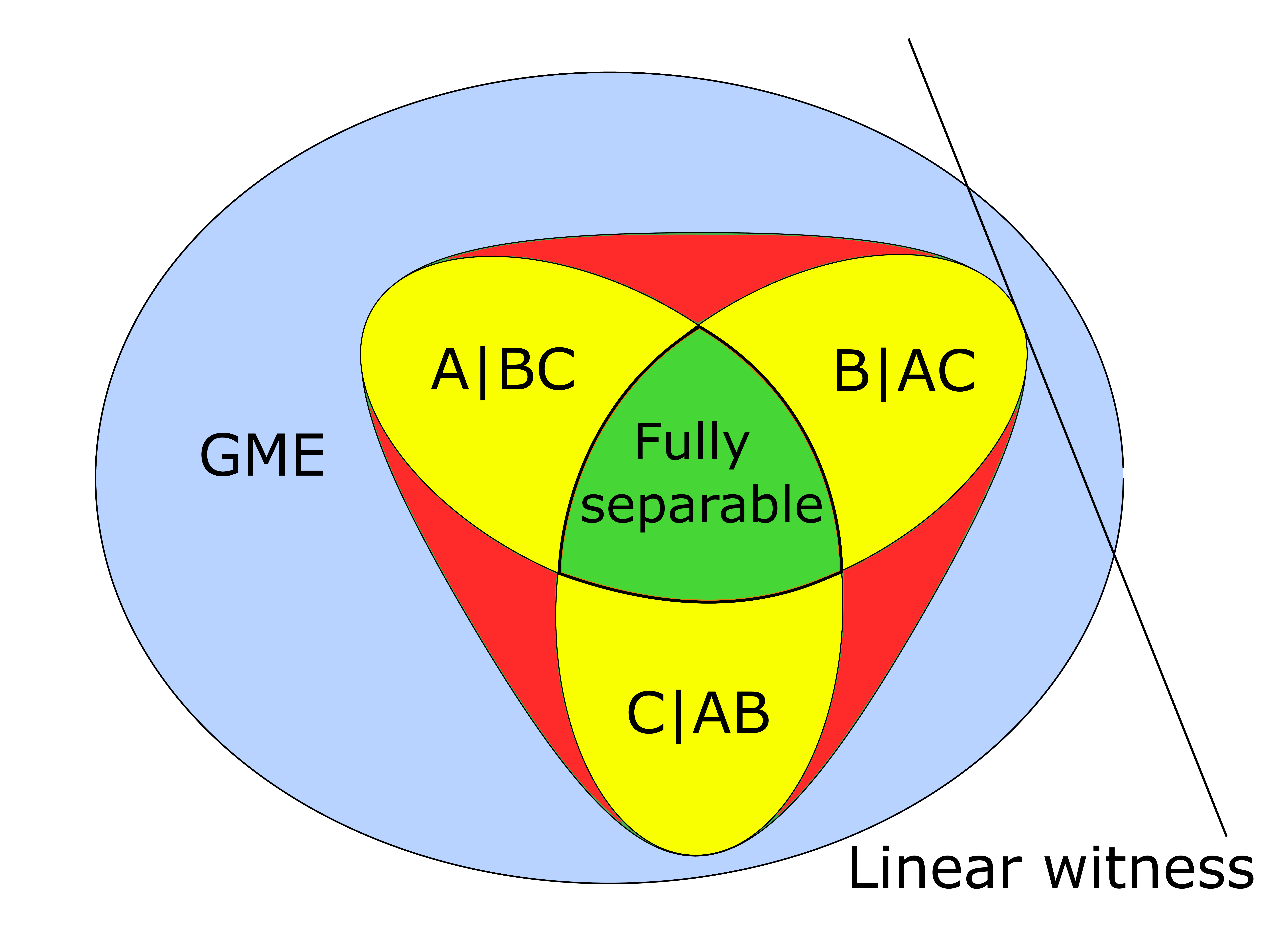}
\caption{(Color online) Schematic representation of the set of tripartite states, adapted from Ref. \cite{setfigure}. In blue is represented the set of GME states. In red is highlighted the set of biseparable states that are not separable with respect to any fixed partition, whereas in yellow are represented the sets of states that are separable with respect to a fixed partition. In green is represented the set of fully separable states. A linear witness defines a hyperplane in the space of states. A non-zero conference key rate can be seen as a non-linear entanglement witness, as it can detect states in the red area, i.e. outside a non-convex set.}
\label{sepset}
\end{figure}

\noindent\emph{Conclusions --} We addressed the question of whether GME is a necessary resource for a conference key agreement protocol. We proved that, surprisingly, the parties can establish a conference key by sharing biseparable states in each round of the protocol. To show this, we exhibited a family of suitable biseparable states, which lead to non-zero key rates in the simple N-BB84 conference key agreement protocol. We showed that, in a network with bottleneck, the key rates achieved by our family of states outperform protocols based on a concatenation of bipartite QKD especially for high numbers of entangled parties.

Furthermore, we related our results to the concept of entanglement witnesses, showing that a non-zero asymptotic conference key rate can only be obtained if one is able to detect entanglement, across any partition, in the state shared by the parties in each round of the CKA protocol. This extends the result of Ref. \cite{Curty2004} for bipartite QKD to the multipartite scenario. As a consequence, we can infer that a non-zero asymptotic conference key rate represents a non-linear entanglement witness, which can detect a type of entanglement that  cannot be detected by the traditional linear entanglement witnesses.

Given our results, several lines of research can follow. For example, it is known that distillation of GHZ states starting from biseparable states is possible \cite{distreview}. Moreover, the GHZ state can be used to generate a perfect conference key. It is an open question whether the considered class of CKA protocols is equivalent to the distillation of a GHZ state from biseparable states. Such a result can lead to converse bounds on the key rates achievable by different classes of multipartite entangled states in the considered CKA protocols.

\section*{Acknowledgements}
We thank F. Grasselli for helpful discussions, and S. Das, S. B\"auml, M. Winczewski and K. Horodecki for clarifying discussions about of the apparent contradiction of our results with Ref.~\cite{das2019universal}. We also thank an anonymous referee for valuable comments that inspired us to strengthen our results. \noindent This work was funded by the Deutsche Forschungsgemeinschaft (DFG, German Research
Foundation) under Germany's Excellence Strategy - Cluster of Excellence
Matter and Light for Quantum Computing (ML4Q) EXC 2004/1 - 390534769.

\bibliography{Bibliography}
\bibliographystyle{apsrev4-2}

\pagebreak
\widetext
\begin{center}
\textbf{\large Supplemental Material: Genuine multipartite entanglement is not a precondition for secure conference key agreement}
\end{center}

\setcounter{equation}{0}
\setcounter{figure}{0}
\setcounter{table}{0}
\setcounter{page}{1}
\makeatletter
\renewcommand{\theequation}{S\arabic{equation}}
\renewcommand{\thefigure}{S\arabic{figure}}

\section{Conference key rate of the N-BB84 protocol with the family of states $\rho_{AB_1, \dots, B_{N-1}}^{(N,k)}$}
As a first step, we briefly sketch the N-BB84 protocol, introduced in Ref. \cite{Federico2018}. The protocol consists of the following steps:
\begin{enumerate}
    \item  A source distributes a state to the $N$ parties.
    \item The parties perform two type of measurements: for the parameter estimation rounds, they make measurements in the $X$ basis. For the key generation rounds they make measurements in the $Z$ basis.
    \item The parties compute the following parameters:
    \begin{itemize}
        \item Using  the outcomes of the parameter estimation rounds the parties compute 
        \begin{equation}
            Q_X=\frac{1-\langle X^{\otimes N}\rangle}{2}
        \end{equation}
        where $\langle X^{\otimes N}\rangle$ is the expectation value of the operator $X$ for each party. $Q_X$ represents the probability that the parties obtain an unexpected result from the parameter estimation rounds.
        \item Using some of the outcomes of the key generation estimation rounds the parties compute
        \begin{equation}
            Q_{AB_i}=\frac{1-\langle Z_{AB_i} \rangle}{2}
        \end{equation}
        where $\langle Z_{AB_i} \rangle$ is the expectation value of the operator $Z$ for Alice and Bob $B_i$. This parameter is computed for each Bob and represents the probability that Alice and Bob $B_i$ get a discordant outcome in the key generation rounds.
    \end{itemize}
    \item The asymptotic key rate is given by
    \begin{equation}
        r_{\rm N-BB84}^{\infty}=1-h(Q_X)-\max_i{h(Q_{AB_i})}
    \end{equation}
    where $h(x)=-x\log_2(x)-(1-x)\log_2(1-x)$ is the binary entropy.
\end{enumerate}

In order to evaluate the performance of the family of states $\rho_{AB_1, \dots, B_{N-1}}^{(N,k)}$, Eq.(7), we need to evaluate the two parameters $Q_X$ and $Q_{AB_i}$. It can be straightforwardly seen that $\rho_{AB_1, \dots, B_{N-1}}^{(N,k)}$, for all $N$ and $k$, is invariant under the application of the $X$ operator on all parties. This implies $\langle X^{\otimes N}\rangle=1$ and thus $Q_X=0$ for any $N$ and $k$.

To calculate $\langle Z_{AB_i} \rangle$ we remark that $\rho_{AB_1, \dots, B_{N-1}}^{(N,k)}$ is a mixture of  $\mathcal{N}= \binom{N-1}{k-1}$ terms, where each of these terms is a projector onto the GHZ state shared by Alice and $k-1$ Bobs, and a projector onto the $|+ \rangle$-state for the remaining Bobs. It is straightforward to see that $\langle Z_{AB_i} \rangle=0$ for the terms in which Bob $B_i$ is not entangled with Alice. On the other hand, the terms in which Bob $B_i$ shares part of the GHZ state with Alice are invariant under the application of the $Z$ operator on Alice and Bob $B_i$, so that we obtain $\langle Z_{AB_i} \rangle=1$ for these terms. Overall, the expectation value $\langle Z_{AB_i} \rangle$ reads
\begin{equation}
    \langle Z_{AB_i} \rangle = \frac{f}{\mathcal{N}}= \frac{k-1}{N-1}.
\end{equation}
where $f=\binom{N-2}{k-2}$ is the number of terms in which Bob $B_i$ shares part of a GHZ state with Alice. We remark that, due to the symmetry of the state, this result holds for any Bob. Thus, dropping the index $i$ we obtain
\begin{equation}
    Q_{AB}(N,k)=\frac{N-k}{2(N-1)}.
\end{equation}
With further, straightforward calculations we obtain
\begin{equation}\label{rateBB84}
    r_{N-BB84}^{\infty}(N,k)=1-h(Q_{AB})=\frac{1}{2}\frac{N-k}{N-1}\log_2{\left(\frac{N-k}{N-1}\right)}+
\frac{1}{2}\frac{N+k-2}{N-1}\log_2{\left(\frac{N+k-2}{N-1}\right)}.
\end{equation}

\section{The N-BB84 protocol is optimal for $Z$ measurements}
In this section we prove that the N-BB84 protocol is the optimal protocol for the family of states $\rho_{AB_1, \dots, B_{N-1}}^{(N,k)}$, when the parties use the $Z$ basis for key generation. We prove this by analyzing the general class of protocol presented in the introduction of the manuscript, thus assuming full state characterization. We show that the key rate of the N-BB84 protocol is identical to the one obtained assuming full tomography of the states $\rho_{AB_1, \dots, B_{N-1}}^{(N,k)}$, thus proving that the N-BB84 protocol is optimal for this family of states.

\subsection{Conditional entropy $H(X|E)$ for the family of states $\rho_{AB_1, \dots, B_{N-1}}^{(N,k)}$}
\label{app1}

Here we will calculate the conditional entropy $H(X|E)$ for a generalization of the family of states $\rho_{AB_1, \dots, B_{N-1}}^{(N,k)}$, as we consider states of the form
\begin{align}
\begin{split}\label{stateq}
\rho_{AB_1, \dots, B_{N-1}}=&\\
\sum_{\underset{S_\alpha \in \mathcal{S}^{(k)} }{\alpha}} q_\alpha & \Phi^{GHZ,k}_{S_{\alpha}} \bigotimes_{\underset{B_m \in \bar{S}_{\alpha} }{m}} | + \rangle \langle + |_{B_m}
\end{split}
\end{align}
where $\Phi^{GHZ,k}_{S_{\alpha}}=|GHZ \rangle\langle GHZ |_{S_{\alpha}}$, is the projector of the GHZ state shared by the parties of the subset $S_\alpha$, defined as $|GHZ \rangle_{S_\alpha}=\frac{1}{\sqrt{2}} \left( |0 \rangle^{\otimes k} + | 1 \rangle^{\otimes k}\right)$ and $|+\rangle = \frac{1}{\sqrt{2}}\left(|0 \rangle + |1 \rangle \right)$. We substituted $\frac{1}{\mathcal{N}}$ with some general real coefficients $q_{\alpha}$ such that $q_\alpha \geq 0 \; \forall \alpha$ and $\sum_\alpha q_\alpha =1$.

We start the explicit calculation of the conditional entropy $H(X|E)$ by writing a purification of the state in Eq. \eqref{stateq}. An explicit valid purification of the state is given by
\begin{equation}
| \psi_{AB_1, \dots, B_{N-1} E} \rangle = \sum_{\underset{S_\alpha \in \mathcal{S}^{(k)} }{\alpha}} \sqrt{q_{\alpha}} | GHZ \rangle_{S_{\alpha}} \bigotimes_{\underset{B_m \in \bar{S}_{\alpha}}{m}} | + \rangle_{B_m} | e_{\alpha} \rangle
\end{equation}
where $\{|e_\alpha \rangle\}_{\alpha}$ is an orthonormal basis of Eve's subsystem of proper dimension. We thus look at the state after Alice performs her measurements on the Pauli $Z$ basis. We obtain the following explicit expression of the state
\begin{align}
\begin{split}
{\rho}_{XB_1, \dots, B_{N-1} E}&=\\ \sum_{\underset{S_\alpha \, , S_\beta \in \mathcal{S}^{(k)} }{\alpha, \beta}}& \frac{1}{2}  \sqrt{q_\alpha q_\beta}  \left(  | 0 \rangle_X \langle 0 | \bigotimes_{B_m \in I_{\alpha,\beta}} | 0 \rangle_{B_m} \langle 0 | \bigotimes_{B_r \in \bar{U}_{\alpha,\beta}} | + \rangle_{B_r} \langle + | \bigotimes_{B_t \in S_\alpha \backslash I_{\alpha,\beta}} | 0 \rangle_{B_t} \langle + | \bigotimes_{B_l \in S_\beta \backslash I_{\alpha,\beta}} | + \rangle_{B_l} \langle 0 | \otimes |e_\alpha \rangle \langle e_\beta | \right. \\
&\quad \quad \left. +  | 1 \rangle_X \langle 1 | \bigotimes_{B_m \in I_{\alpha,\beta}} | 1 \rangle_{B_m} \langle 1 | \bigotimes_{B_r \in \bar{U}_{\alpha,\beta}} | + \rangle_{B_r} \langle + | \bigotimes_{B_t \in S_\alpha \backslash I_{\alpha,\beta}} | 1 \rangle_{B_t} \langle + | \bigotimes_{B_l \in S_\beta \backslash I_{\alpha,\beta}} | + \rangle_{B_l} \langle 1 | \otimes |e_\alpha \rangle \langle e_\beta | \right)
\end{split}
\end{align}
where $I_{\alpha,\beta}= (S_\alpha \cap S_\beta)$ is the intersection and $U_{\alpha,\beta}= S_\alpha \cup S_\beta$ the union between the subsets of the Bobs in $S_\alpha$ and $S_\beta$, $\bar{U}_{\alpha,\beta}$ is the complement of $U_{\alpha,\beta}$ and ${\rho}_{XB_1, \dots, B_{N-1} E}$ indicates the state after Alice's measurement. We can then trace out all the Bobs, which leaves us with Alice and Eve's reduced state in the form
\begin{eqnarray}
\nonumber
\rho_{XE}= \sum_{\alpha,\beta} \frac{1}{2}  \frac{\sqrt{q_\alpha q_\beta}}{2^{k-s_{\alpha,\beta}}} | 0 \rangle_X \langle 0 | \otimes | e_\alpha \rangle \langle e_\beta |+ \sum_{\alpha,\beta} \frac{1}{2}  \frac{\sqrt{q_\alpha q_\beta}}{2^{k-s_{\alpha,\beta}}} | 1 \rangle_X \langle 1 | \otimes | e_\alpha \rangle \langle e_\beta | = \\
= \sum_{\alpha,\beta} E_{\alpha,\beta}  \frac{1}{2}(| 0 \rangle_X \langle 0 | + | 1 \rangle_X \langle 1 |) | \otimes | e_\alpha \rangle \langle e_\beta | = \frac{\mathds{1}_X}{2} \otimes \rho_E
\end{eqnarray}
where $s_{\alpha,\beta}$ is the cardinality of $I_{\alpha,\beta}$, where we defined $E_{\alpha,\beta}=\frac{\sqrt{q_\alpha q_\beta}}{2^{k-s_{\alpha,\beta}}}$ in the second line of the equation and where $\rho_E= \sum_{\alpha,\beta} E_{\alpha,\beta} |e_\alpha \rangle \langle e_\beta |$ is Eve's reduced state. Finally, since $\rho_{XE}$ is a product state, we can use Property 1 of the conditional entropy to write $H(X|E)=H(X)=1$, thus concluding the proof.

\subsection{Conference Key rates for the family of states $\rho_{AB_1, \dots, B_{N-1}}^{(N,k)}$}
\label{app2}

We now evaluate the analytical expression for the asymptotic key rate for the family of biseparable states $\rho_{AB_1, \dots, B_{N-1}}^{(N,k)}$, given by Eq. (7) in the main text. We recall that the number of terms in the convex combination is equal to the number of subsets of cardinality $k-1$ within the $N-1$ Bobs, which is equal to $\mathcal{N}=\binom{N-1}{k-1}$, and that we consider all the coefficients to be equal to $q_{\alpha}=\frac{1}{\mathcal{N}}$. 

To calculate the asymptotic key rate, since $H(X|E)=1$, as proven in Section \ref{app1}, we need to evaluate the leakage $H(X|Y_i) \; \forall$ Bob$_i$ which, with our choice of coefficients, will be equal for all the Bobs. We thus calculate the reduced density matrix of Alice and Bob$_i$ after they perform the key generation measurements, $\rho_{XY_i}$, in order to estimate the leakage term.
Tracing out all the Bobs except one and performing the measurement both on Bob$_i$ and Alice's side gives us the state 
\begin{equation}
\rho_{XY_i}=\frac{1}{2}\frac{f}{\mathcal{N}}(|0 \rangle_X \langle 0 | \otimes |0 \rangle_{Y_i} \langle 0 |+ |1 \rangle_X \langle 1 | \otimes |1 \rangle_{Y_i} \langle 1 |) + (1-\frac{f}{\mathcal{N}})\frac{\mathds{1}_{XY_i}}{4},
\end{equation}
where $f$ is the number of terms in which Bob$_i$ is entangled with Alice in the original state. The number $f$ can be expressed in term of $k$ and $N$ as $f=\binom{N-2}{k-2}$. Thus the reduced density matrix in the computational basis has the form
\begin{equation}
\rho_{XY_i} = 
\left[
\begin{array}{cccc}
\frac{1}{4}(1+C_{N,k})  & 0 & 0 & 0 \\
0 &\frac{1}{4}(1-C_{N,k}) & 0 & 0 \\
0 & 0 &   \frac{1}{4}(1-C_{N,k}) & 0 \\
0 & 0 & 0 & \frac{1}{4}(1+C_{N,k})
\end{array} \right] ,
\end{equation}
where $C_{N,k}=\frac{f}{\mathcal{N}}=\frac{k-1}{N-1}$. Note that the reduced density matrix of Bob$_i$ after the measurement is $\rho_{Y_i}= \frac{\mathds{1}_{Y_i}}{2}$. We therefore obtain
\begin{eqnarray}
r_{\infty}(N,k) & = & 1-H(XY_i)+H(Y_i)  \nonumber \\
& = & \frac{1}{2}\frac{N-k}{N-1}\log_2{\left(\frac{N-k}{N-1}\right)}+ \frac{1}{2}\frac{N+k-2}{N-1}\log_2{\left(\frac{N+k-2}{N-1}\right)}. 
\end{eqnarray}
The key rate obtained with this method is equivalent to Eq. \eqref{rateBB84}, thus proving that the N-BB84 protocol is optimal for $Z$-basis measurements for the key generation rounds.

\section{Noise analysis for the N-BB84 protocol}
In this section we consider a noise model for the N-BB84 protocol with the family of states $\rho_{AB_1, \dots, B_{N-1}}^{(N,k)}$ and compare its performance with a concatenation of bipartite QKD protocols between Alice and $N-1$ Bobs, where all the channels between Alice and the Bobs are noisy. For a fair comparison we thus consider local depolarizing noise. This corresponds to applying the map 
\begin{equation}\label{whitenoise}
    \mathcal{D}[\rho]=(1-p)\rho + p \frac{\mathds{1}}{2}
\end{equation}
to each of the Bobs. The state we will consider will thus be
\begin{equation}\label{noisestate}
    \rho_{AB_1, \dots, B_{N-1}}^{noise}=\mathcal{D}^{\otimes (N-1)}[\rho_{AB_1, \dots, B_{N-1}}^{(N,k)}]
\end{equation}
In this scenario, the parameters of the N-BB84 protocol can be analytically evaluated and read
\begin{eqnarray}
Q_x & = & \frac{1-(1-p)^{N-1}}{2} \\
Q_{AB} & = & \frac{N-1-(1-p)(k-1)}{2(N-1)}
\end{eqnarray}
where, again, we dropped the index $i$ since, due to the symmetry of the state, all $Q_{AB_i}$ are equal. We thus can evaluate analytically the key rate for the N-BB84 protocol, which reads
\begin{align}
    r_{N-BB84}^{\infty}(N,k,p)  =  \frac{1}{2}(1-(1-p)^{N-1})\log_2{(1-(1-p)^{N-1})}+\frac{1}{2}(1+(1-p)^{N-1})\log_2{(1+(1-p)^{N-1})} \nonumber \\  +  \frac{N-1-(1-p)(k-1)}{2(N-1)}\log_2{\left( \frac{N-1-(1-p)(k-1)}{2(N-1)}\right)} \\
     +  \frac{N-1+(1-p)(k-1)}{2(N-1)}\log_2{\left( \frac{N-1+(1-p)(k-1)}{2(N-1)} \right)}
\end{align}

We compare it with the scenario where Alice performs a bipartite BB84 protocol with each of the Bob, sharing a maximally entangled state mixed with white noise, as in Eq. \eqref{whitenoise}. The resulting key rate of a concatenation of bipartite BB84 protocols reads \cite{BB84proof,Epping2017}
\begin{equation}\label{QKD}
    r_{QKD}^{\infty}(N)=\frac{1-2h(\frac{p}{2})}{N-1}
\end{equation}
where we divide the key rate of the bipartite BB84 protocol in the presence of white noise by the number of times Alice must perform the bipartite protocol in order to establish a secure key with each of the $N-1$ Bobs. The results are shown in Figure \ref{noise}.
\begin{figure}[!h]
\centering
{\includegraphics[height=8cm, width=8cm, keepaspectratio]{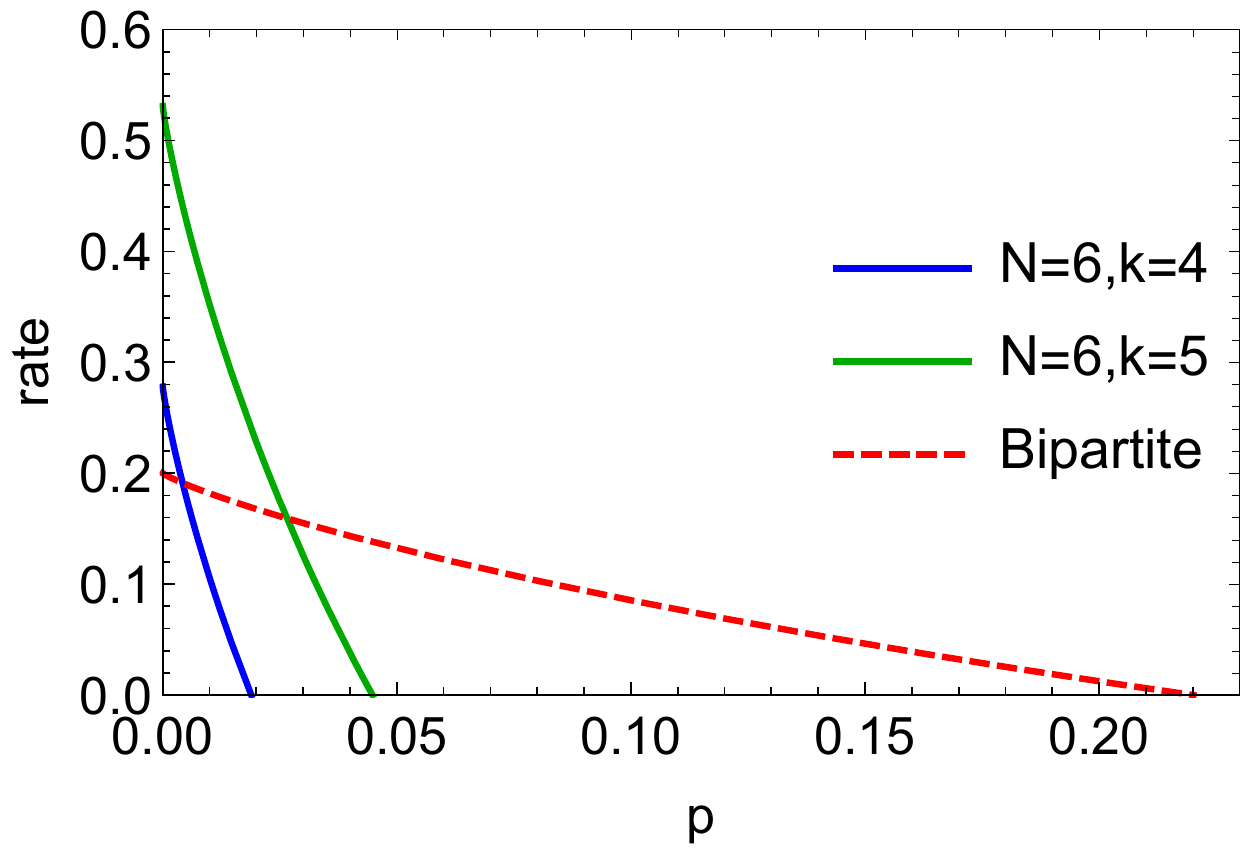}} \quad
{\includegraphics[height=8cm, width=8cm, keepaspectratio]{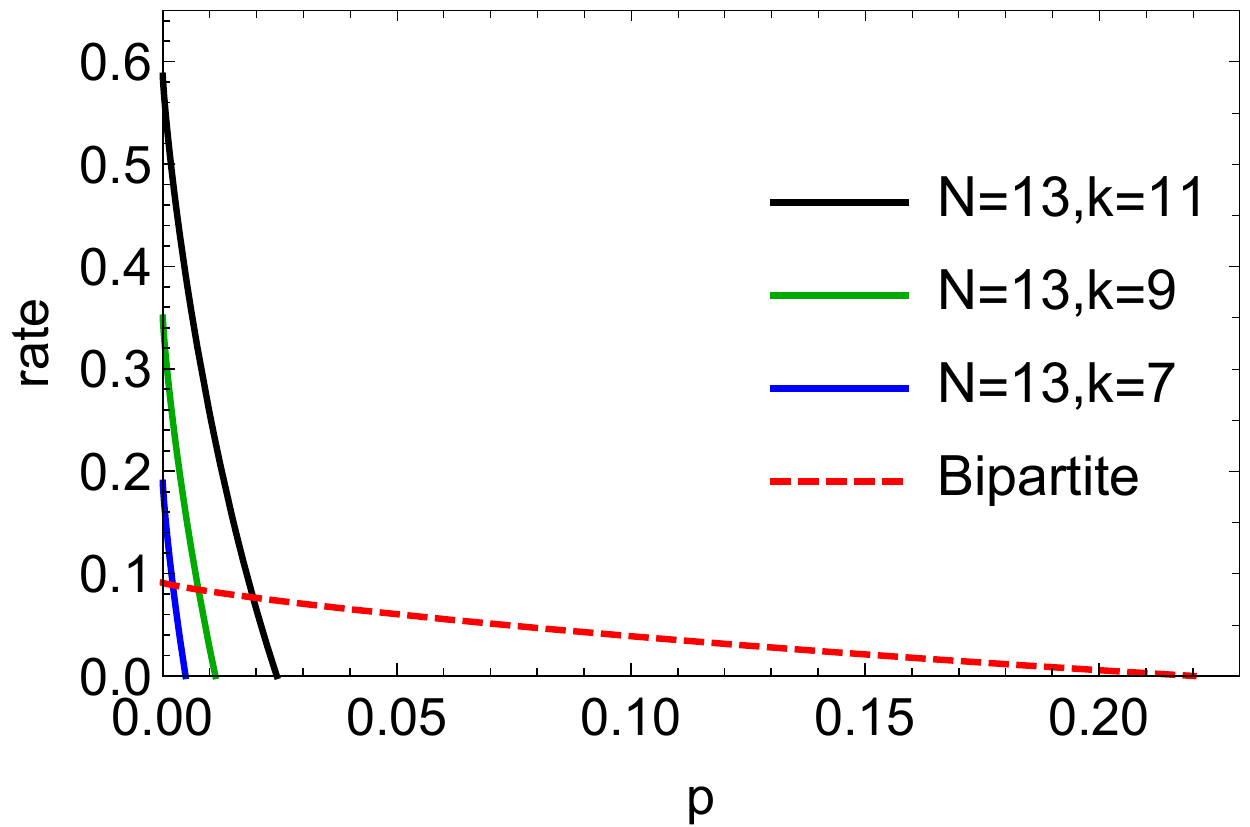}}
\caption{Left panel: plot of the asymptotic key rate of the N-BB84 protocol for the state of Eq. \eqref{noisestate} (solid lines) as a function of $p$, for fixed $N=6$ and different $k$: $k=4$ (blue, left) and $k=5$ (green, right). The results are compared with the key rate of a concatenation of BB84 QKD protocols, given in Eq. \eqref{QKD} (red dashed line), for $N=6$, as a function of $p$. Right panel: plot of the asymptotic key rate of the N-BB84 protocol for the state of Eq. \eqref{noisestate} (solid lines) as a function of $p$, for a fixed value of $N=13$ and different values of $k$: $k=7$ (blue, left), $k=9$ (green, middle) and $k=11$ (black, right). The results are compared with the key rate of a concatenation of BB84 QKD protocols, given in Eq. \eqref{QKD} (red dashed line), for $N=13$, as a function of $p$.}
\label{noise}
\end{figure}
We can see that for some regimes the N-BB84 protocol outperforms a concatenation of bipartite QKD protcols: for low number of parties, we can obtain a marked advantage for $k$ close to $N$ in the low noise regime. Moreover, increasing the number of parties increases the advantage obtained and the range of $k$ for which we can obtain it. However, we note that the N-BB84 protocol has a lower noise tolerance than the concatenation of bipartite QKD protocols, and thus for high noise regimes the latter is always preferred. 

\section{Proof of Theorem 2}
We give here the full proof of Theorem 2. For completeness we repeat the statement of the theorem.

\newtheorem*{thm2}{Theorem 2}

\begin{thm2}
Given a CKA protocol in which the parties use a set of local measurements, for the test  and key generation rounds, which are represented by the POVMs $\{G^a_{x}\}, \{G^{b_1}_{y_1}\}, \dots, \{G^{b_{N-1}}_{y_{N-1}}\}$, where $a,b_1, \dots b_{N-1}$ indicate the outputs of the measurements labelled by $x,y_1,\dots,y_{N-1}$, 
then one can obtain a non-zero asymptotic conference key rate $r_\infty > 0$ only if the presence of entanglement can be proved across any partition of the parties.

Moreover, the presence of entanglement across each partition can be verified through a set of entanglement witnesses of the form
\begin{equation}
W_\alpha=\sum_{\underset{a, b_1, \dots , b_{N-1}}{x,y_1,\dots,y_{N-1}}} c_{\underset{a, b_1, \dots , b_{N-1}}{x,y_1,\dots,y_{N-1}}}^{(\alpha)} G_{x}^{a} \otimes G^{b_1}_{y_1} \otimes \dots \otimes G^{b_{N-1}}_{y_{N-1}}
\end{equation} 
where $\alpha$ labels the partition $S_\alpha | \bar{S}_\alpha$ with $S_\alpha$ being a proper subset of the parties and $\bar{S}_\alpha$ is its complement, and where $c_{\underset{a, b_1, \dots , b_{N-1}}{x,y_1,\dots,y_{N-1}}}^{(\alpha)}$ are real coefficients.
\end{thm2}

\begin{proof}
We start by focusing on the probability distribution of the outcomes $a,b_1, \dots, b_{N-1}$ given the inputs $x, y_1, \dots, y_{N-1}$ of the measurements that can be performed in the test and key generation rounds of the CKA protocol, namely $P(a,b_1, \dots, b_{N-1} | x, y_1, \dots, y_{N-1})$. The probability distributions are obtained as
\begin{equation}
\label{defprob}
P(a,b_1, \dots, b_{N-1} | x, y_1, \dots, y_{N-1})= \mbox{Tr}(G_{x}^{a} \otimes G^{b_1}_{y_1} \otimes \dots \otimes G^{b_{N-1}}_{y_{N-1}} \rho_{AB_1 \dots B_{N-1}}),
\end{equation} 
where $G_x^a$, $G_{y_i}^{b_i}$ are the POVM elements of the measurements performed by Alice and Bob$_i$, respectively.

We analyze the map that maps each state into the corresponding probability distribution, given the measurements of the protocol, that is
\begin{align}
\Pi_{CKA}: \rho_{AB_1 \dots B_{N-1}} \mapsto \{ P(a,b_1, \dots, b_{N-1} | x, y_1, \dots, y_{N-1}) \}
\end{align}
Considering a subset of the Hilbert space, namely $\Sigma$, we call $\Sigma^{\Pi}$ the projection of the subset $\Sigma$ through the map $\Pi_{CKA}$, defined as in Eq. \eqref{defprob}. We now denote the set of states separable across the partition $S_\alpha | \bar{S}_\alpha$ as $\Sigma_\alpha$. We note that $\Sigma_\alpha$ is a closed and convex set. Furthermore, the projection of the set $\Sigma_\alpha$ through the linear map $\Pi_{CKA}$, namely $\Sigma_\alpha^{\Pi}$ is still a closed and convex set. The elements of the projected set represent the probability distributions that come from states that are separable across the partition $S_\alpha | \bar{S}_\alpha$. Due to Theorem \ref{teo}, a necessary condition to obtain a non-zero key rate is that the state is not separable with respect to any partition. This implies that, given a state $\rho^*_{A,B_1, \dots, B_{N-1}}$ that leads to a non-zero key rate in a specific protocol, the corresponding probability distribution $P^*(a,b_1, \dots, b_{N-1} | x, y_1, \dots, y_{N-1})$ is such that  $P^*(a,b_1, \dots, b_{N-1} | x, y_1, \dots, y_{N-1}) \notin \Sigma^{\Pi}_\alpha \; \forall \alpha$. Moreover, since each $\Sigma^{\Pi}_\alpha$ is a convex and compact set, it is a well known fact that each element of its complement $\bar{\Sigma}_\alpha^{\Pi}$ can be separated from $\Sigma^{\Pi}_\alpha$ with a proper hyperplane \cite{HB,witness2}. In the probability space any hyperplane can be defined as
\begin{equation}
\sum_{\underset{a, b_1, \dots , b_{N-1}}{x,y_1,\dots,y_{N-1}}} c_{\underset{a, b_1, \dots , b_{N-1}}{x,y_1,\dots,y_{N-1}}} P(a,b_1, \dots, b_{N-1} | x, y_1, \dots, y_{N-1}) = 0
\end{equation}
where $c_{\underset{a, b_1, \dots , b_{N-1}}{x,y_1,\dots,y_{N-1}}}$ are real coefficients. Furthermore, for each probability distribution $P^*(a,b_1, \dots, b_{N-1} | x, y_1, \dots, y_{N-1}) \notin \Sigma^{\Pi}_\alpha \; \forall \alpha$, we can find, for each partition $S_\alpha | \bar{S}_\alpha$,  coefficients $c^{(\alpha)}_{\underset{a, b_1, \dots , b_{N-1}}{x,y_1,\dots,y_{N-1}}}$, defining hyperplanes such that 
\begin{align}
\label{witprob}
\forall P_\alpha(a,b_1, \dots, b_{N-1} | x, y_1, \dots, y_{N-1})  \in & \Sigma^{\Pi}_\alpha \; \; \sum_{\underset{a, b_1, \dots , b_{N-1}}{x,y_1,\dots,y_{N-1}}} c^{(\alpha)}_{\underset{a, b_1, \dots , b_{N-1}}{x,y_1,\dots,y_{N-1}}} P_\alpha(a,b_1, \dots, b_{N-1} | x, y_1, \dots, y_{N-1}) \geq 0 \nonumber \; \; \mbox{and} \\
\text{for } P^*(a,b_1, \dots, b_{N-1} | x, y_1, \dots, y_{N-1})  \notin & \Sigma^{\Pi}_\alpha \; \forall \alpha, \;\; \sum_{\underset{a, b_1, \dots , b_{N-1}}{x,y_1,\dots,y_{N-1}}} c^{(\alpha)}_{\underset{a, b_1, \dots , b_{N-1}}{x,y_1,\dots,y_{N-1}}} P^*(a,b_1, \dots, b_{N-1} | x, y_1, \dots, y_{N-1}) < 0
\end{align}
Finally, the coefficients define a set of entanglement witnesses in the form
\begin{equation}
W_\alpha=\sum_{\underset{a, b_1, \dots , b_{N-1}}{x,y_1,\dots,y_{N-1}}} c_{\underset{a, b_1, \dots , b_{N-1}}{x,y_1,\dots,y_{N-1}}}^{(\alpha)} G_{x}^{a} \otimes G^{b_1}_{y_1} \otimes \dots \otimes G^{b_{N-1}}_{y_{N-1}}
\end{equation}
such that, due to Eq. \eqref{witprob}, for each $\alpha$
\begin{align}
\label{witne}
\mbox{Tr}(W_\alpha \sigma_\alpha) &\geq 0 \,,\; \forall \sigma_\alpha \in \Sigma_\alpha \nonumber \\
\mbox{Tr}(W_\alpha \rho^*_{A,B_1, \dots, B_{N-1}})& < 0.
\end{align}

As a matter of fact, Eq. \eqref{witne} tells us that the operator $W_\alpha$ is an entanglement witness \cite{distreview,witness2} that detects entanglement across  partition $S_\alpha | \bar{S}_\alpha$. This concludes the proof of the Theorem.
\end{proof}

\end{document}